\documentclass{article}

\usepackage{fancyhdr}
\date{}  

\usepackage{url}
\usepackage{cite}
\usepackage{plain}

\usepackage{wrapfig}
\usepackage{graphics}
\usepackage{picins,graphicx}
\usepackage[english]{babel}
\usepackage[leqno]{amsmath}
\usepackage{amssymb}
\usepackage{verbatim}
\usepackage[mathscr]{eucal}

\usepackage{amstext}
\usepackage{makeidx}
\usepackage{amsthm}

\usepackage[bookmarks,colorlinks]{hyperref}
\hypersetup{colorlinks=false}

\usepackage[bookmarks,colorlinks]{hyperref}
\hypersetup{colorlinks=false}

\usepackage{hyperref}
\hypersetup{colorlinks,%
citecolor=red,%
filecolor=black,%
linkcolor=blue,%
urlcolor=blue}

\makeindex

\newtheorem{theorem}{Teorema} 

\newtheorem{proposition}{Proprieta'}
\newtheorem{definition}{Definizione}
\newtheorem{notation}{Nota}
\newtheorem{ex}{Esercizio} 
\newtheorem{esempio}{Esempio}

\newcommand{\beq}{\begin{equation}} 
\newcommand{\eeq}{\end{equation}}

\newcommand{\bex}{\begin{ex}} 
\newcommand{\eex}{\end{ex}} 

\newcommand{\bese}{\begin{esempio}} 
\newcommand{\eese}{\end{esempio}} 

\newcommand{\bpro}{\begin{proposition}} 
\newcommand{\epro}{\end{proposition}}

\newcommand{\bthe}{\begin{theorem}} 
\newcommand{\ethe}{\end{theorem}}

\newcommand{\bnote}{\begin{notation}} 
\newcommand{\enote}{\end{notation}}

\newcommand{\bdefi}{\begin{definition}} 
\newcommand{\edefi}{\end{definition}} 

\newcommand{\bc}{\begin{center}} 
\newcommand{\ec}{\end{center}}

\newcommand{\mail}[1]{\href{unina:#1}{\texttt{#1}}}

\usepackage{pdfsync}

\author{Monica De Angelis\thanks{Univ. di Napoli  "Federico II", Scuola Politecnica e delle Scienze di Base, Dip. Mat. Appl. "R.Caccioppoli", \newline
 Via Cinthia 26, 80126, Napoli , Italia.
\newline\mail{modeange@unina.it}}}

\title{A wave equation perturbed by viscous terms: fast and
slow times diffusion effects in a Neumann problem}
\begin{document}

\title{A wave equation perturbed by viscous terms: fast and
slow times diffusion effects in a Neumann problem
}

\author{Monica De Angelis\thanks{Univ. di Napoli  "Federico II", Scuola Politecnica e delle Scienze di Base, Dip. Mat. Appl. "R.Caccioppoli", \newline
 Via Cinthia 26, 80126, Napoli , Italia.
\newline\mail{modeange@unina.it}}}

\author{Monica De Angelis \\Univ. di Napoli  "Federico II",\\ Scuola Politecnica e delle Scienze di Base,\\ Dip. Mat. Appl. "R.Caccioppoli", \\
 Via Cinthia 26, 80126, Napoli , Italia.
\\ \mail{modeange@unina.it}
}


\maketitle

A Neumann problem for a wave equation perturbed by        viscous terms with  small parameters is considered. The 
interaction of  waves with the diffusion effects caused by  a higher-order derivative with small coefficient $ \varepsilon $, is investigated. Results obtained prove that for slow time $ \varepsilon t <1 $  waves are  propagated almost undisturbed,  while for fast time $ t>\frac{1}{\varepsilon•} $ diffusion effects prevail.

 \section{Introduction}
\label{intro}

The paper deals with an analysis of the diffusion equation:

\begin{equation}                                   \label{11}
 {\cal L}_ \varepsilon u\equiv \partial_{xx}(\varepsilon
\partial_t+1 )u - \partial_t(\partial_t+a)u=
f
 \end{equation}

\noindent
which, as it happens for  artificial viscosity methods, represents a
model of wave equation perturbed by viscous terms with  small
parameters.

 According to the meaning of source term $ f, $  equation (\ref{11}) represents numerous examples of dissipative phenomena   in several fields  such  as physics, neurobiology or engineering, and in each of these subjects an  extensive bibliography exists in this regards. (see,f.i. \cite{rf}-\cite{drultimo} and reference therein).

 In particular, for  $f=\sin u + \gamma $,  (\ref{11}) expresses the perturbed
Sine-Gordon equation  which models the Josephson effect in
Superconductivity \cite{bp}. In  many other cases, extra terms can  be considered  to describe  various  Josephson junctions such as semiannular, S-shaped,  window  or exponentially shaped.(see,f.i \cite{bc}-\cite{32}).

Moreover,  an equivalence  between the third order equation (\ref{11}) and an integro differential equation has been proved in \cite{m13}-\cite{mda10}  and this further  allows us  to  create a direct connection between biological phenomena   and superconductivity. For instance,  the propagation of  nerve impulses  described by the  FitzHugh - Nagumo model  can be related to the Pertubed Sine  Gordon Equation. (see\cite{acscott},too). 

Naturally,  equation (\ref{11}) is to be complemented by initial problems  and boundary conditions  such as Neumann, Dirichlet, pseudo periodic or mixed problems    that are meaningful in many scientific fields (see,f.i.\cite{m13,df13}). Particular attention will be given to Neumann problem that is relevant in the ecological model when the exterior environment is completely hostile to the species \cite{m1}, or in the study of cardiac rhythmicity,  particularly for  pacemakers \cite{ks}. Moreover, in superconductivity, in the case of a Josephson junction, Neumann problem can refer to  the phase gradient value that is proportional to the magnetic field    \cite{for,j05}.

\subsection{Mathematical considerations, state of the art and aim of the paper}

When the behaviour of
$u (x,t)$ as $\varepsilon \rightarrow 0 $ is examined, the 
interaction of  waves with the diffusion effects caused by  $\varepsilon u_{xxt}$  can be estimated, and this physical aspect is meaningful to
the evolution of dissipative models.

When a fixed  boundary-initial problem  ${\cal P}
_\varepsilon$ is stated and the function $f$ is linear,  by means of the related  Green function
$G_\varepsilon, $ it is possible  to solve   ${\cal P}
_\varepsilon$ explicitly. Otherwise, if the source term $f$ is non linear, the problem  ${\cal P}
_\varepsilon$ can be reduced to an  integral equation
\cite{dmm}.

So, for $\varepsilon \equiv 0$, the parabolic equation (\ref{11})
 turns into the following equation

\begin{equation}              \label {12}
{\cal L} _ 0 U  \equiv
    (
 \partial_{xx}  -a\partial_t-\partial_{tt}) U  = \ f,\ \  \
\end{equation}

\noindent
and  ${\cal P}
_\varepsilon$ changes into a problem   ${\cal P}
_0$ for $U$, with the {\em same initial-boundary conditions} of  ${\cal P}
_\varepsilon.$ 

In  small time intervals, the wave behaviour is a  believable  approximation of $ u_\varepsilon  $ when $ \varepsilon $ is vanishing. Conversely, when the time $ t $ is large, diffusion effects  should prevail and the behaviour of    $ u_\varepsilon  $ when $ \varepsilon \rightarrow \,0  $
and $ t \rightarrow \, \infty $ should  be analyzed.

So that, denoting by $ T $ an arbitrary positive constant, let

\begin{center}
$\Omega =\{(x,t) : 0  \leq x \leq
\pi, \  \ 0 < t \leq T \}$,
\end{center}

\noindent and let us assume as ${\cal P}
_\varepsilon$ the following linear  Neumann problem  :

  \begin{equation}                                                     \label{13}
  \left \{
   \begin{array}{ll}
     \partial_{xx}(\varepsilon
u _{t}+ u) - \partial_t(u_{ t}+a)=f(x,t),
      & (x,t)\in \Omega,\\ \\
   u(x,0)=f_0(x), \qquad    u_{t}(x,0)=f_1(x), \  \ & x\in [0,\pi],
   \\ \\
     u_x(0,t)=\varphi(t), \qquad u_x(\pi,t)=\psi(t), \ & 0<t \leq T.
   \end{array}
  \right.
 \end{equation}


It is important to underline that problem (\ref{13}) defined on  space interval $ [0,\pi] $  is    equivalent to  the system defined on an arbitrary interval  $ [0,L]. $ Indeed,  it is possible to consider a 
finite interval  $ [0,L] $ by  rescaling $ t\longrightarrow  \tau \,c$  and $ x \longrightarrow    c\,\bar x \, $ with $ c= \pi/L. $  So that, in  many cases  the spatial coordinate $ x  $ is normalized  without loosing generality (see, f.i. \cite{bp},\cite{ddf},\cite{fd05}).

Moreover, as $f(x,t),
\,\,f_0(x), \,\,f_1(x)$ are quite arbitrary, it is not restrictive to assume $\varphi(t)=0, \ \psi(t)=0$. Otherwise, it suffices to put

\begin{equation}  \label {14}
\bar{u} = u -  \frac{x}{2\pi}\,[\,\,( 2\pi-x) \,\varphi \,+ \,x\,\,\psi\,\,]
\\ \\\end{equation}
\[ F = f+ \biggl(\frac{\varepsilon}{\pi}+ \frac{\alpha x^2}{2\pi}\biggr)\biggl(\dot \varphi -  \dot\psi \biggr)+ \frac{1}{\pi}\biggl( \varphi+\psi \biggr) +\frac{x^2}{2\pi}\biggl( \ddot \varphi -\ddot \psi \biggr)-x(  \alpha  \dot \varphi+\ddot \varphi)\] 

\noindent and then modify $f_0(x), f_1(x)$ accordingly. So, henceforth removing the superscripts,   one obtains:

  \begin{equation} \label{15}
  \left \{
   \begin{array}{ll}
     \partial_{xx}(\varepsilon
u _{t}+ u) - \partial_t(u_{t}+a)=F(x,t),
      & (x,t)\in \Omega,\\ \\
   u(x,0)=F_0(x), \qquad    u_t(x,0)=F_1(x),  & \ x\in [0,\pi],
   \\ \\
     u_x(0,t)=0, \qquad \qquad  u_x(\pi,t)=0, \ & \ 0<t \leq T
   \end{array}
  \right.
 \end{equation}

 \noindent where 
 \begin{equation}  \label {16}
 \left \{
   \begin{array}{ll}
 F_0   = f_0 - \frac{x}{2\pi}\,\,[ \,(2\pi-x)\, \varphi(0)\,\,+\,x\, \psi(0)\,] ;
\\ \\ F_1   = f_1 - \frac{x}{2\pi}\,\,[ \,(2\pi-x) \,\dot \varphi(0)\,\,+\,x\,\dot  \psi(0)\,].
\end{array}
  \right.
\end{equation}

Consequently, the problem  $ {\cal P}_0  $ is the following:

  \begin{equation} \label{17}
  \left \{
   \begin{array}{ll}
     U_{xx} - U_{tt}-aU_t =F(x,t),
      & (x,t)\in \Omega,\\ \\
   U(x,0)=F_0(x), \qquad    U_t(x,0)=F_1(x),  &\ x\in [0,\pi],
   \\ \\
     U_x(0,t)=0, \qquad \qquad  U_x(\pi,t)=0, \ & \ 0<t \leq T.
   \end{array}
  \right.
 \end{equation} 

Let us define $ G_\varepsilon(x,\xi,t) $ as  the Green function  of the operator $ {\cal L}_\varepsilon = \varepsilon \partial_{xxt} + \partial_{xx} - \partial_{tt}-a \partial _t $   introduced in (\ref{11}). It is possible to   determine $ G_\varepsilon  $ by means of  Fourier series. Indeed, letting :

\begin{equation}                                             \label{18}
 h_n=\frac{1}{2}(a+\varepsilon n^2),\quad  \omega_n=\sqrt{h_n^2-n^2}; \quad  G^\varepsilon_n(t)= \frac{1}{\omega_n}e^{-h_nt}
\sinh(\omega_nt),
\end{equation}

\noindent it results:

\begin{equation}                                                \label{19}
G_\varepsilon (x,t,\xi)= \frac{1}{\pi•}\,\, \frac{1-e^{-a\,\,t}}{a• }\,\,+\,\,\frac{2}{\pi}\,\,\sum_{n=1}^{\infty}
G^\varepsilon_n(t) \,  \, \cos (n\xi)\,  \, \cos(nx).
\end{equation}

Moreover , assuming 

\begin{equation}                                             \label{110}
 \omega_0=\sqrt{n^2-(a/2)^2} \qquad G^o_n(t)= e^{-\frac{at}{2}}\,\,\frac{1}{\omega_0}
\sin(\omega_0t), \\ \
  \end{equation}

the Green function  related to problem (\ref{17}) is:

\begin{equation}                                                \label{111}
G_o(x,t,\xi)= \frac{1}{\pi}\,\, \frac{1-e^{-a\,\,t}}{a• }\,\,+\frac{2}{\pi}\, \,\sum_{n=1}^{\infty}
G^o_n  \,\,\,\cos (n\xi)\,  \, \cos (nx).
\end{equation}

As for the state of art, the interactions between diffusion effects and wave propagation  have already been studied for Dirichlet conditions. In particular,  when  $ a=0, $ an asymptotic approximation for the Green function  has been  established by means of the  two characteristic times: slow time $ \tau = \varepsilon \,t $  and fast time $ \theta =\, t/\varepsilon $ \cite{dmr}. Moreover, for the  semilinear  equation related to an exponentially shaped Josephson junction, in \cite{df213}  an analytical  analysis   proves      that the  surface damping has little influence
on the behaviour of  oscillations,  thus  confirming numerical results showed  in  \cite{bcs00}. Other numerical investigations on  the influence of surface losses can be found in \cite{pskm},too.

As for the operator $ {\cal L}_\varepsilon, $ in \cite{uffa} it is possible to find a   short review, while   for the  Dirichlet and Neumann  problem in   \cite{mda}-\cite{df13} the Green function  has  already been 	explicitly examined  by means of a Fourier series,   
and its properties  have been proved.

Aim of this paper is to provide an estimate of  the influence of the the diffusion term         $ \varepsilon u_{xxt} $ on wave propagation.  Parameters $ 0< a <1  $ and $ 0<\varepsilon <1 $ are adopted, and,  according to problems (\ref{15}) and (\ref{17}), the following difference

\begin{equation}  \label{112}
|u(x,t,\varepsilon)- e^{-\frac{\varepsilon\, \,t}{2 •}}\,\, U(x,t)|
\end{equation}

\noindent is evaluated proving, as expected, that  $ u(x,t,\varepsilon)$ tends to $ U(x,t) $ as soon as  $ \varepsilon  $ tends to zero. 
 
Moreover , when  slow times $ \varepsilon t  <1$  are  considered, there exists a positive constant $ K  $ depending on data and source term, such that

\begin{equation}  \label{113}
\bigg| u(x,t,\varepsilon) - e^{\frac{\varepsilon \,t}{2•}\, } U(x,t)\bigg | < K. 
\end{equation}

 \noindent Hence,  for $ t\in (0,\frac{1}{\varepsilon•}), \,$ the wave  is propagated almost undisturbed. 
 
  Conversely, for fast time  
  $ \varepsilon t  >1,$  damped oscillations 
due to term $ \varepsilon u_{xxt} $ predominate.

The estimate allows us also to evaluate the behaviour   when time $ t $  tends to infinity  proving that, for a suitable source term,  solution $ u(x,t,\varepsilon)$  is bounded.

The paper is organized as follows:

In section 2 some properties of  the Green function $ G_\varepsilon(x,t,\varepsilon) $ are recalled and  the solution  related to problem (\ref{15}) is determined.

In section 3 an analysis  of the difference  $ \big |G_\varepsilon (x,t,\varepsilon ) - e^{\frac{\varepsilon  \,\,t }{2•}} G_0(x,t)\big|,  $ and its derivative with respect  to $ t, $ is performed. Besides, a further inequality for Green Function   is achieved.

Finally, in section 4,  difference  (\ref{112})  is estimated, and  behaviours due to slow and fast times are underlined together with asymptotic behaviour when time $ t $ increases.

 \section{On  properties related to Green function  $ G_\varepsilon $ and explicit solution of the third order problem }

In \cite{mda12,mda} some properties  of the  Green function for Dirichlet-type boundary conditions  have already been determined. It is possible to prove that the same properties are  equally valid   for the Green function related to a Neumann boundary problem.

So that, denoting by  $ \beta  \equiv  \min \,\,\bigl\{ (\varepsilon+a)^{-1}
;   \ {(a+\varepsilon)
/2}; a \bigr\}, $ and letting

\begin{equation} \label{daa}
G:= \frac{2}{\pi•}\sum_{n=1}^{\infty}   \,\,G^\varepsilon_n(t)\, \cos(n\xi)\,  \, \cos(nx),\quad   
\end{equation}

\noindent since (\ref{19}), it results: 

\begin{equation} \label{da}
G_\varepsilon (x,t,\xi)= \frac{1}{\pi•}\,\, \frac{1-e^{-a\,\,t}}{a• }\,\,+\, G(x, \xi,t)   
\end{equation}
and,according to Lemma 2.1  of \cite{mda}(assuming $\ell    \equiv \pi \,\,\mbox{and}\,\, c^2=1  $), one has:

\begin {theorem}

\label{th1}
For $ a,\varepsilon  \in\Re^+$,  the function $G_\varepsilon(x,\xi,t)$  defined in
(\ref{19}) and all its time derivatives are continuous functions. Moreover, there exist  positive constants  $ M, $ and     $ D_j \,\, (j \in {\sf N})$   depending on   $a $ and $\varepsilon,$ such that:

\begin{equation} \label{217}
 \left |  G(x,\xi,t)\right| \leq M e^{-\beta
t}     \qquad \left|  \frac{\partial ^j G_\varepsilon}{\partial t^j} \right|  \, \leq \,D_j \,  e^{-\beta
t}, \quad j \in {\sf N}  
\end{equation}

\end{theorem}


Besides,    as for x - differentiation, denoting by $ G_{\varepsilon,t} :=\frac{\partial G _{\varepsilon}}{\partial t}$ and $G_{\varepsilon,x} :=\frac{\partial G _{\varepsilon}}{\partial x}, $ the  uniform convergence of  $ G_{\varepsilon,x} $  is assured by means of standard  criteria,  while   the absolute convergence of  $ (\varepsilon\, G_{\varepsilon,t}\, +\, G_\varepsilon )  $ and its first and second  derivatives  are  	guaranteed by means of the following inequalities already proved in  Lemma 2.2 of \cite{mda}. Indeed, one has:

\begin{theorem} \label{th2}
For all $a,\,
\varepsilon, $ $\in$ $\Re^+$,
the function $G_{\varepsilon,x}(x,\xi,t)$  is a  continuous function  and it converges uniformly for all $x\in [0, \pi]$. Moreover, it results

\begin{equation}  \label{219}
  |\partial_{x}^{(i)}\,\,(\varepsilon\, G_{\varepsilon,t}\, +\, G_\varepsilon )| \,\leq M_i \, \,
\, e^{-\beta t},\qquad (i=1,2,3) 
\end{equation}

\noindent where $ M_i \,\, (i=1,2,3)$ are all positive constants depending on $ a \,\mbox{and} \,\,\varepsilon.  $

\end{theorem}

\noindent Furthermore, according to Theorem 2.1 of \cite{mda}  it results: 

\begin{equation}                      \label{218}
{\cal L}_\varepsilon G_\varepsilon \, =\partial_{xx}(\varepsilon
G _{\varepsilon,t}+ G_{\varepsilon}) - \partial_t(G_{\varepsilon,t}+a G_\varepsilon)=0.
\end{equation}

Now, let  $f(x) $  be  a continuous function on $(0,\pi),$  and let

\begin{eqnarray}  \label{220}
 \tilde {f} \,:  = \frac{1}{\pi} \int_{0}^{\pi} f(\xi) \,  \,\
 d\xi    
   \end{eqnarray} 
Then, the following theorems hold:

\begin{theorem} \label{th:33}
Let $F_1(x)$ be  a  $C^1$ function on $ [0,\pi] \, $ with  $ \,\dot F_1(0)=\,\dot F_1(\pi)=0.
\,$ 
The function

\begin{equation} \label{221}
u_{1}(x,t)=
\biggl(\frac{1- e^{-a \, t}}{a} \biggr)\tilde F_1 \,\ +\,\int_{0}^{\pi} F_1(\xi) \ \ G(x,\xi,t)\ \ d\xi\
\end{equation}

 \noindent is a solution
 of the
equation ${\cal L}_\varepsilon  u_{1} =0$ and it satisfies the  homogeneous boundary conditions. Moreover it  results:

\begin{equation}    \label{222}
 \lim_{t \rightarrow 0} u_{1}(x,t)=0, \ \ \ \ \ \lim_{t \rightarrow 0} \partial_t
u_{1}(x,t)=F_1(x)
\end{equation}

\noindent  uniformly for all $x\in [0, \pi]$.
\end{theorem}

\begin{proof}
Theorems \ref{th1} and  \ref{th2}  and continuity of $ F_1 $  assure
that function
$(\ref{221})$ and its partial derivatives  converge absolutely for all
$(x,t)\in \Omega$ and   ${\cal L}_\varepsilon \,u_{1} =0$.
Besides,  $(\ref{222})_1$ holds, too.

Furthermore, since:

\begin{equation}  \label{224}
\partial_t  \int_{0}^{\pi} F_1(\xi) \ \ G(x,\xi,t)\ \ d\xi =- \, \frac{2}{\pi} \int_{0}^{\pi}\sum_{n=1}^{\infty}
\ \dot G^\varepsilon_n \,\,\,\dot F_1(\xi) \, \frac{\sin (n\xi)}{n} \,\cos
(n x) \ d\xi\,,
\end{equation}

\noindent denoting by $\eta(x)$the Heaviside function, it results:

\begin{equation}                            \label{226}
\lim_{t \rightarrow 0}  \partial_tu_{1}=
\tilde F_1+ \int_{0}^{\pi} \,\,
\bigl[\frac{\xi}{\pi}- \eta(\xi-x)\bigr]\,\,\dot F_1(\xi)
d\xi  =\,
F_1(x).
\end{equation}

\noindent Moreover, by means of the uniform convergence of $ G_x  $, also   boundary conditions $(\ref{15})_3$ hold as well. 

\end{proof}

\begin{theorem} \label{the:dati0}

If $F_0(x) \in C^2[0,\pi]$  with
$ \dot F_0(0)=\,\dot F_0(\pi)=0
\,$, then the function 

\begin{equation}
u^*(x,t)\, = \,\,\tilde F_0 \,\,+ \,
(\partial_t+a-\varepsilon\partial_{xx}) \int_{0}^{\pi} F_0(\xi) \ \ G(x,\xi,t)\ \ d\xi 
\end{equation}

\noindent  is a solution   of the
equation ${\cal L}_\varepsilon u^* =0$, it  satisfies  boundary conditions   and it results:

\begin{equation} \label{228}
\lim_{t \rightarrow 0}u^*(x,t)=F_0(x), \ \ \ \ \ \lim_{t \rightarrow 0} \partial_t
u^*(x,t)=0,
\end{equation}

\noindent
{\em uniformly for all }$x\in [0,l]$.

\end{theorem}

\begin{proof}  Let us define: 
\begin{equation}
u_{F_0}:= \int_{0}^{\pi}  F_0(\xi) \ \ G(x,\xi,t)\ \ d\xi  ; \qquad u_{\ddot F_0}:=\int_{0}^{\pi} \ddot F_0(\xi) \ \ G(x,\xi,t)\ \ d\xi  
\end{equation}
 Hypotheses  on $F_0(x)$ and theorem \ref{th1}  
 assure that $  \partial_{xx}u_{ F_0}= u_{\ddot F_0} $
  and so,  since (\ref{218}), equation  ${\cal
L}_\varepsilon\, u^* =0$ is verified as well. Moreover, being 
\begin{equation}
 ( \partial_{tt} \,+ a \partial_{t}) u_{ F_0} = ( \varepsilon \partial_{xxt} + \partial_{xx})u_{ F_0} \,
 \end{equation} 

\noindent it results $ \,\, \partial_t u^* =u_{\ddot F_0}\,\,$ and so 
 $(\ref{228})_2 $ holds.
Finally,  being

\begin{equation}   \label{231}
\lim_{t \rightarrow 0}u^*= \tilde F_0 \, +\,\lim_{t \rightarrow 0} \partial _t u_{ F_0} \,
\end{equation}

\noindent similarly to  (\ref{224})-(\ref{226}), ($\ref{228})_1$ follows, too. 
\end{proof}

Now, let  us consider the convolution of the Green function with the source term F(x,t). 
\begin{theorem}   \label{t42} 
Let the function $F(x,t)$ be  a continuous function
in $\Omega$ with continuous derivative with respect to x,  
 then the function 
 
 \begin{equation} \label{an217}
 u_F\,=- \int_0^t \biggl (\frac{1- e^{-a (t-\tau) }}{a} \biggr)\tilde F(\tau)  d\tau -\int_0^t d \tau \int_{0}^{\pi} F(\xi,\tau) G(x,\xi,t-\tau)d\xi,
 \end{equation}
 
\noindent satisfies equation $ {\cal L}_\varepsilon \,u_F  = F\, $ and   homogeneous boundary conditions are verified. Moreover, one has:

\noindent  
\begin{equation}    \label{an219}
 \lim_{t \rightarrow 0} u_F(x,t)=0, \ \ \ \ \ \lim_{t \rightarrow 0} \partial_t
u_F(x,t)=0,
\end{equation}

\noindent  uniformly for all $x\in [0, \pi]$.
\end{theorem}
\begin{proof}
In the same way as  $(\ref{224}) $ and  $(\ref{226}),$ it results:

\[ \lim_{\tau \rightarrow t} \,\,\bigg( 
\tilde F(t) + \int_{0}^{\pi} F(\xi,\tau) \ \ G_t(x,\xi,t-\tau)\ \ d\xi\ \bigg)\,= F , \]

\noindent so that, one has:

\begin{equation} \label{an223}
   \partial_{tt}\,u_F(x,t) =  - F(x,t) + a \int_{0}^{t} \tilde F  e^{-a (t-\tau)}d\tau\,
 -\int_{0}^{\pi} F(\xi,\tau)  G_{tt}(x,\xi,t-\tau) d\xi 
\end{equation}

\noindent and, by means of properties   (\ref{218}), ${\cal L}_\varepsilon \,u_F  = F\,$ is verified.

\noindent
Furthermore,
owing to estimates (\ref{217}), since (\ref{an217}),  initial  homogeneous conditions follow.

Moreover, by means of the uniform convergence of $ G_x  $,   boundary conditions $(\ref{15})_3$ hold, as well.   
\end{proof}

Since uniqueness can be a consequence of the energy method (see, f.i.\cite{mda} and reference therein), the following theorem holds:

 \begin{theorem}
When data $ (F_1,F_0,F) $ satisfy  respectively the hypotheses of theorems \ref{th:33},                      \ref{the:dati0},  \ref{t42}, then

\begin{equation}  \label{an225}
 u(x,t)=\,\,\tilde F_0 \,+\,\, \biggl(\frac{1- e^{-a \, t}}{a} \biggr)\tilde F_1 \,\,+ 
\int_{0}^{\pi} \,\,F_1(\xi)\,\, G(x,\xi,t) d\xi\,\,+ \,\,
\\ \\
\end{equation}  
\[ 
\,\, \, (\partial_t+a -\varepsilon\partial_{xx})\int_{0}^{\pi} F_0(\xi)
\,\,G(x,\xi,t) d\xi -\, \int_0^t \,\,\bigg(\frac{1- e^{-a \, (t-\tau) }}{a} \biggr)\tilde F(\tau) \, d\tau\ \]
\[  - \int_0^ t d\tau\, \int_0^\pi\, G(x,\xi,t-\tau)\,
F(\xi,\tau)d\xi. \]

\noindent
 represents the unique solution of problem \ref{15}.

 \end{theorem}

\section{ Estimates related to the Green Functions }

Let us consider the  following difference:

\begin{equation}                \label {44}
 \sum ^\infty _{n=1} \bigg[ e^{-h_n \,\,t} \, \frac{\sinh (\omega_n \,t)}{ n^2 \,\omega_n} \,-\, e^{-h_1 t} \,  \frac{\sin (\omega_0 \,t)}{ n^2 \,\omega_0}\bigg] g_n  =  \sum ^\infty _{n=1} H_n g_n \,
\end{equation}

\noindent where   
\begin{equation}
\,\, g_n = \frac{2}{\pi•}\cos(n\xi) \cos(nx);\,\quad h_1= (a+\varepsilon )/2 .\
\end{equation}

Assuming $ 0< a <1  $ and $0< \varepsilon <1, $ let us denote by $ N_1 $ the minimum  natural number larger than $\frac{1}{ \varepsilon}\,\,(\, 1 - \sqrt{ 1-  a \, \varepsilon }\,), $ and by $ N_2  $ the maximum natural number smaller than  $\frac{1}{ \varepsilon}\,\,(\, 1 + \sqrt{ 1-  a \, \varepsilon }\,).\,\, $  

Since   $ N_1 <1, $   $ G^\varepsilon_n(t) $  in $ (\ref{18})_3 $ contains trigonometric functions for  $1 \leq n\, \leq\, N_2 $  and hyperbolic terms for $ n\geq N_{2}+1. $  Hence  intervals $ (1,N_2); (N_{2}+1,\infty)$ must be considered.

So, denoting 

\begin{equation}  \label{41}
\sum ^\infty _{n=1} H_n g_n=  \sum ^{N_2} _{n=1}  \,  H_n ( t,\varepsilon ) \, g_n + \sum^\infty _{n=N_2+1}  \,  H_n ( t,\varepsilon ) \, g_n = H_1 +H_2, 
\end{equation}

\noindent the following theorem holds:
\begin{theorem} \label{th7}
 
Whatever    $ 1/2<\gamma<1,$ and $ 0<\delta<2$ may be,  there exists a positive constant $ A $, independent from $ \varepsilon, $ such that the following estimate holds:

\begin{equation}  \label{38}
\bigg|\sum_{n=1}^{\infty} [\,G_n^\varepsilon(t)-  e^{\frac{-\varepsilon  t}{2}} 
 G_n^o(t)\,]\frac{g_n}{n^2}\,\bigg| \leq A [  \varepsilon^{1-\gamma•} \, r(t) \, e^{-\,\frac{a}{2}\, t} + \varepsilon  e^{-\,\frac{1}{4}\, \theta}]  
 \end{equation}

\noindent  where  $ \theta   $ denotes the fast time $ t/\varepsilon $  and $r(t) =1+t+t^{1-\gamma}+t^{2-\delta}\,.$

\end{theorem}

\begin{proof}

 Referring to the trigonometric terms related to $H_1$ defined in (\ref{41}), indicating by $ \hat H_n (s,\varepsilon) $   the Laplace transform of function $ H_n(t,\varepsilon ), $  one deduces that  

\[n^2\hat H_n(s, \varepsilon)= \frac{(s+h_1)^2•+\omega_0^2 -(s+h_n)^2•-\tilde \omega_n^2}{[(s+h_1)^2•+\omega_0^2] \,\,\,[(s+h_n)^2•+\tilde\omega_n^2}
 \]
 
\noindent where  $ \tilde \omega_n^2 \,\, = {n^2- h_n^2 }. $  So that, denoting by $ \tilde g_1= \frac{1}{2}(a+\frac{\varepsilon}{2})  ,
$   it results

\begin{eqnarray} \label{399}
 &\nonumber H_n( t,\varepsilon) = \frac{ \varepsilon\tilde g_1   }{n^2 \,\omega_0 \tilde \omega_n }   \int _0^t  e^{-h_1(t-\tau)} \sin (\omega_0 (t-\tau))  e^{-h_n \tau}   \sin (\omega_n \,\tau) d\tau  +
 \\
 \\
 &\nonumber\frac{ \varepsilon   (n^2-1)}{n^2\omega_o•} \int_0^t    e^{-h_1(t-\tau)} \sin (\omega_0 (t-\tau))  \{ e^{-h_n \tau} [\frac{h_n}{\tilde\omega_n•}  \sin (\tilde\omega_n  \tau)-  \cos (\tilde\omega_n \,\tau)]\}\,.
\end{eqnarray}

\noindent\[ 
\nonumber \hspace{3mm}\mbox{ Now,  letting }\qquad \qquad g_0=  \sqrt{1-\frac{a^2}{4}}, \quad \,g_1= \frac{1}{2•}(a+\frac{1}{2•}), \qquad \qquad\]

 \noindent one has
 
\begin{equation}  \label{370}
  \omega_0 \geq n\, g_0, \quad \, \tilde g_1 \leq g_1.  
\end{equation}  
Moreover, since  $ \tilde \omega_n\,\geq  n \sqrt{1-h_n/n•}= n  \,
\Phi_n, \, $  indicating by

\begin{eqnarray}
 \nonumber & s =  \sqrt{\frac{2\sqrt{1-a\varepsilon}-\varepsilon•}{2(1-\varepsilon+\sqrt{1-a\varepsilon})}} = \frac{( \Phi_n )_{n= N_2-1}•}{\sqrt{\varepsilon }•}, \quad  q= \sqrt{\frac{2-a-\varepsilon }{2 }}=( \Phi_n )_{n=1}\\ \nonumber\\
 \nonumber \mbox { and by}&\ell= \min \,\{ \,s, q\,\},\qquad
  \end{eqnarray}

 \noindent   it is possible to choose  a positive constant  $\, g_2 \, $ depending on parameter  $ a, $ but  independent from $ \varepsilon, $  such that $ g_2 \leq  \ell\,.\,  $ In such a way, for all  $  1\leq n\leq N_2-1\,, $ one has:

\begin{equation}  \label{40}
 \tilde \omega_n\, \geq n \,\sqrt{\varepsilon} \,\, g_{2}.      
\end{equation}

 Therefore, taking into account that   $\, h_n-h_1 =\varepsilon  (n^2-1)/2 $ and that  
 
 \begin{equation} \label{a}
 e^{-x} \leq \frac{\alpha^\alpha •}{(ex)^\alpha •} \qquad  \qquad \forall \,\,x\, > 0; \, \,\,\forall \alpha >0\,,
 \end{equation}  
 
  \noindent considering  as well that  $ \varepsilon^2 (N_2^2-1)\,\geq    1-a\,+2\sqrt{1-a} :=\rho \,,$  from (\ref{399}) it results:

\begin{equation}
 e^{\frac{a\,t}{2•} }\sum_{n=1}^{N_2} H_n \leq   \sum_{n=1}^{N_2-1}\bigg ( \frac{g_1}{ n^4 g_0 g_2•} + \frac{a}{2 g_0g_2n^2•}    \bigg) \sqrt{\varepsilon}  t + 
\end{equation}

\[\bigg[ \sum_{n=1}^{N_2}\frac{ 1 }{n g_o•}  + \sum_{n=2}^{N_2-1}\frac{ \varepsilon ^{1/2}  }{ g_0g_2•} \bigg]\frac{(2 \gamma)^\gamma}{(e )^\gamma (1-\gamma)} \,\,\frac{(\varepsilon t )^{1-\gamma}}{(n^2-1)^{\gamma}} +\]

\[ + \bigg[ \frac{ g_1  \varepsilon ^4}{ \,g_0 (1+\sqrt{1-a\varepsilon})^3•}  + \frac{ \varepsilon^2   a+\varepsilon  (\sqrt{1-a\varepsilon})^2}{2g_0(1+\sqrt{1-a\varepsilon})•}\bigg] \frac{2 (\delta/e)^{\delta}•}{ \rho^\delta•(2-\delta)}\, \varepsilon ^\delta t^{2-\delta} ,\]

\noindent where     it is assumed that $ 1/2<\gamma<1, $  and $ 0<\delta<2. $

 So    it is possible to choose a positive constant $A_1\,, $   independent from $ \varepsilon,  $  such that:  
 
\begin{equation}\label{422}
H_1 =\sum_{n=1}^{n=N_2} H_n g_n \leq  A_1 [\sqrt{\varepsilon} t  +(\varepsilon t )^{1-\gamma} +\varepsilon ^{1+\delta} t^{2-\delta}] e^{-\frac{a}{2•}t}.
\end{equation}

  As for $H_2$ of (\ref{41}), one has:

\begin{equation} \label{42}
 H_2  =   \sum ^\infty _{n=N_2+1} e^{-h_n \,\,t}  \frac{\sinh (\omega_n \,t)}{ n^2 \,\omega_n}    g_n - e^{-h_1t} \sum ^\infty _{n=N_2+1} \frac{sen(\omega_0 t)}{n^2\, \omega_0}  g_n  = H'_2-H''_2.
\end{equation}

Being

\begin{equation} \label{46}
e^{-h_n t} \,\,\, \frac{\sinh (\omega_n \,t)}{  \,\omega_n}= e^{-(h_n-\omega_n)t }\int _0^t e^{-2\omega_n\tau}  d\tau \,\,\leq  t \,\,e^{-(h_n-\omega_n)t }
\end{equation}

and
\begin{equation} \label{47}
h_n-\omega_n = h_n - h_n \sqrt{1-\frac{n^2}{h_n^2•}•} \geq  h_n - h_n \bigg(1-\frac{n^2}{2 h_n^2•}\bigg),
\end{equation}

since $ n\geq N_2+1 > 1/ \varepsilon $ and $ a\varepsilon <1, $ one has:

\begin{equation} \label{48}
\,e^{-(h_n-\omega_n)t }\leq  \,\, e^{- t\,\,\frac{n^2}{2 h_n•}•}\leq  e^{- \,\frac{t}{2\varepsilon }}\,.
\end{equation}

 As  a consequence, denoting by  $\,\,\theta $ the fast time $ t/\varepsilon $ and by  $ \zeta(z)   $   the Riemann zeta function, it results:

\begin{equation} \label{44}
\mid H'_2\mid   \leq \,\, t  \,\,  \, e^{-\frac{t}{2\varepsilon}}\,  \sum ^\infty _{n=1} \,\, \frac{1}{n^2}\,\, = \zeta(2) \,\, \varepsilon \,\, \theta\,\, e^{-\,\frac{1}{2}\, \theta}.
\end{equation}
\end{proof}

Finally, since $(\ref{370})_1$ one deduces that:

\begin{equation}                             \label{416}
\, \mid H''_2 \mid \,\, \leq \,\, \frac{e^{-\varepsilon t/2}}{ g_0} \,\, \sum ^\infty _{n=N_2+1} \,\, \frac{ e^{-at/2} }{n^3}\,\, \leq \frac{1}{  g_0} \,\, \frac{e^{-at/2}}{N_2+1} \zeta(2) \leq \,\,A_2\,\, \varepsilon \,\,e^{-at/2}
\end{equation}

where $ A_2 = (1+\sqrt{1-a\varepsilon} +\varepsilon)^{-1}\,\,\zeta(2)  /g_0  \leq   \zeta(2)  /g_0.  $ 

Referring to ${(\ref{44})},$  by means of (\ref{a}), one has: 

\begin{equation}                   \label{417}
\theta \, e^{-\frac{1}{2}\, \theta} \,\, \leq \,\,( 4 /e
)\,\, 
e^{-\frac{1}{4}\,\theta}\Rightarrow  \mid H'_2\mid   \leq ( 4 /e
)\,\, \zeta(2) \varepsilon \,\, e^{-\,\frac{1}{4}\, \theta}\,.
\end{equation}

So, by means of (\ref{422}),(\ref{416}),(\ref{417}), there exists a positive constant $  A  $   such that inequality (\ref{38}) holds.

Now,   according to solution (\ref{an225}), attention must be paid to the derivative with respect to variable $ t. $ So,  the following theorem is proved:

\begin{theorem} \label{th8}
 Whatever  $ 1/2<\gamma<1,  $ $ 0<k<1/2\, $ and $ 0<\eta<1 $ may be,  there exists a positive constant $C, $ independent from   $ \varepsilon,  $ such that the following estimate holds:

\begin{equation}
   \label{499} 
\nonumber\bigg| \partial _t \bigg[ \sum_{n=1}^{\infty} G^\varepsilon _n(t)\,- \, e^{\frac{-\varepsilon  t}{2}}  \sum_{n=1}^{\infty} 
 G_n^o(t)
  \bigg]\frac{ g_n}{n^2} \bigg| \leq C \,\,\varepsilon^m  \{ e^{-\frac{\theta•}{2•}} +
  p(t)  e^{-\frac{a}{2•} t}\} 
\end{equation}
\noindent   where 
\begin{equation}  \label{51}
m = \min (1-\gamma; \,\,  1-\eta;\,\,1/2-k,); \,\quad  p(t)= 2t+t^{1-\gamma}+3\,,
\end{equation}

\noindent and $ \theta   $ denotes the fast time $ t/\varepsilon. $ 
\end{theorem}

\begin{proof}
As for  terms of $ H_1 $ of (\ref{41}), when $ n= N_2, $ since $ \int _0^t \,\tau\, e^{- \varepsilon (N_2^2-1)  \tau/2}  \,  d\tau \leq \varepsilon /(1-a)\,, $ it is possible to find a  positive constant $ C_1,  $  independent from $ \varepsilon, $  such that, from (\ref{399}) one obtains:

\begin{equation} \label{bb}
e^{\frac{a•}{2•}t}|\partial_t H_1| \leq \varepsilon  \sum_{n=1}^{N_2-1•}  \frac{\omega_0 +h_1}{\omega_0\tilde\omega_n} \bigg(\frac{ \tilde g_1   }{n^2 } +   h_n+\tilde\omega_n\bigg)  \int _0^t  e^{- \varepsilon (n^2-1)  \tau/2}   d\tau  + C_1 \varepsilon (1+t)
\end{equation}

\noindent where, as defined in $(\ref{18})_1$,  $\, h_n= (a+\varepsilon n^2)/2. \,$ 

 Besides, letting  by  $ c_0 $ the  Euler constant, one has  $ \sum_{n=1}^{n=N_2} \,\,\frac{1}{n}\leq  c_0 \, +\ \frac{1}{2N_2}+ ln\,  N_2, $ with $\,\,\ ln \, N_2 \, \leq \,  (\beta /e )\,\, N_2 ^{1/\beta}\,\, \forall \beta >0.   $ So that, assuming $ \beta = k^{-1}\, $ $ (k>0) $ it results:

\begin{equation} \label{53}
\sum_{n=1}^{n=N_2} \,\,\frac{1}{n}\leq      \,\,    \bigg[ c_0 \, +\ \frac{\varepsilon   }{2(1+ \sqrt{1-a\varepsilon)}}+ \frac{  (1+ \sqrt{1-a\varepsilon}) ^k }{ek\,\, \varepsilon ^{k}}\bigg]
\end{equation}

Therefore, since   
$ \int _0^t  e^{- \varepsilon n^2  \tau/2}   \,\,d\tau \leq\frac{2}{\varepsilon n^2•} , $ there exists a positive constant $ C_2, $ independent of $ \varepsilon,  $ such that:

\begin{equation} \label{52}
\sum_{n=1}^{N_2-1}\frac{\varepsilon^2 n^2}{2 \tilde\omega_n•}\int _0^t  e^{- \varepsilon (n^2-1)  \tau/2}   \,\,\,\,d\tau \leq  C_2 \,( t \varepsilon^{3/2} + \varepsilon^{1/2-k}) 
\end{equation} 
where  $ 0<k<1/2. $

Other terms can be evaluated  taking  account of

\begin{equation}
\varepsilon \,\,\, \sum_{n=1}^{N_2-1•}  \, \int _0^t  e^{- \varepsilon (n^2-1)  \tau/2}   \,\,\,\,d\tau  \leq  \varepsilon t +\sum_2^{\infty}  \,\,  \frac{(2 \gamma)^\gamma}{e ^\gamma (1-\gamma)} \frac{ \varepsilon^{1-\gamma}\,t^{1-\gamma}}{(n^2-1)^{\gamma}}
\end{equation}
with $ 1/2 <\gamma<1. $

As for terms of  $  H_2, $ since (\ref{42}), it results:

\begin{eqnarray}
\nonumber |\partial_t H_2| \leq &\sum ^\infty _{n=N_2+1} \frac{e^{-(h_n-\omega_n)t}}{2\,\,n^2\, \omega_n•}  \big( |\omega_n- h_n| + (h_n +\omega_n) e^{-2\omega_nt }\big)+ 
\\
\nonumber &\sum ^\infty _{n=N_2+1}\,\frac{e^{-\frac{a•}{2•}t}}{n^2}\,( \frac{h_1}{\omega_0}+ 1) =   |\partial_t H'_2| +|\partial_t H''_2|.
\end{eqnarray}

\noindent Since   $ \forall n\geq N_2+1  $ function $ \frac{h_n}{\omega_n•} $   decreases, one has  $ \big( \frac{h_n}{\omega_n•}\big)_n \leq \big(\frac{h_n}{\omega_n•}\big)_{N_2+1} = \frac{1}{\sqrt{\varepsilon}•}\,\,\varphi\,\, $   where $ 0< \varphi \leq 9/(2\sqrt{2}) \,.  $  So, taking also into  account   (\ref{48}), it results:

\begin{equation}
 |\partial_t H'_2| \leq \sum ^\infty _{n=N_2+1} \frac{e^{- \,\frac{t}{2\varepsilon }}} {\,\,n^{1+k}\, }\, \frac{\,\,\varepsilon^{1/2-k}\,\,\varphi}{(1+\sqrt{1-a\varepsilon}+\varepsilon•)^{1-k}} 
\end{equation}

Moreover, it results as well:

\begin{eqnarray} \label{557}
|\partial_t H''_2|  \leq \sum ^\infty _{n=N_2+1} \frac{e^{-\frac{a•}{2•}t}}{n^{1+\eta}}\,( \frac{h_1}{\omega_0}+ 1)  \frac{ \varepsilon ^{1-\eta}}{ ( 1+\sqrt{1-a\varepsilon }+  \varepsilon )^{1-\eta}} 
\end{eqnarray}

\noindent with $ 0<\eta<1. $
So that, since (\ref{bb}) (\ref{52})-(\ref{557}), it is possible to find a positive  constant $ C \,, $ independent from $ \varepsilon, $  such that

\begin{eqnarray}
\label{cc}
\nonumber &\bigg| \partial _t \bigg[ \sum_{n=1}^{\infty} G^\varepsilon _n(t)\,- \, e^{\frac{-\varepsilon  t}{2}}  \sum_{n=1}^{\infty} 
 G_n^o(t)
  \bigg]\frac{ g_n}{n^2} \bigg| \leq C\{ e^{- \,\frac{t}{2\varepsilon }} \,\,\varepsilon^{1/2-k}+  
  \\ & e^{-\frac{a•}{2•}t}[ \varepsilon+\varepsilon^{1/2-k}+\varepsilon^{1-\eta}+ (\varepsilon t)^{1-\gamma}+2\varepsilon t]      \}                 
\end{eqnarray}

\noindent from which  (\ref{499}) follows.
\end{proof}

Finally, it must be observed  that, although the Green function $ G_\varepsilon (x,\xi,\varepsilon) $ satisfies theorems (\ref{th1}) and (\ref{th2}), a further theorem for the function $ G, $  defined in (\ref{daa}), must be proved. So that, one has:

\begin{theorem} \label{th9}

Whatever  $ 0< a <1  $ and  $ 0<\varepsilon <1 $  may be, there exists a  positive constant $ B , $   independent from $ \varepsilon,  $ such that:

\begin{equation}\label{588}
\varepsilon\,|   G(x,\xi,t)|\,\leq  \,B \,[\, (1+t)\,\varepsilon  ^{1/2-k} \,e^{- \frac{a}{2•}\,t}\, + \varepsilon^{1-\eta}  \,e^{-  \frac{\theta }{4\,}}\,\,]
\end{equation}

\noindent where $ 0<\eta<1, $   $ 0<k<1/2, $     and    $ \,\,\theta= t/\varepsilon. $

\end{theorem}

Indeed,  for circular terms, since (\ref{40}) and  (\ref{53}), it is possible to find a positive constant $ B_1\,, $ independent from $ \varepsilon, $ such that:   

 \begin{equation} \label{aa}
\varepsilon \,\,\sum_{n=1}^{N_2} \,\frac{ e^{-h_n t}\,sen (\tilde\omega_n \,t)}{\tilde \omega_n } \leq B_1\,\, e^{-\frac{a\,t}{2}} ( \varepsilon^{1/2-k} +t \varepsilon)
\end{equation}

\noindent with $ 0<k<1/2.\, $  As for hyperbolic terms, let   $c$  be an  arbitrary positive constant less than $ 1, $  and  denote  by  $ N_c   $    the integer part of $ 1/( \varepsilon \sqrt{c}) (1+ \sqrt{1-a\varepsilon \,c}). $ So that,  intervals $ (N_{2}+1, N_{c}-1); ( N_{c}, \infty)   $  will be considered.

\noindent Since (\ref{46})-(\ref{48}), letting  $ B_2 = \varepsilon (N_c-N_2-1) \leq \frac{2}{ \sqrt{c}•},$ 
it results: 

\begin{equation} \label{uu}
\sum_{n=N_{2}+1}^{N_c-1}  e^{-h_n t} \,\frac{\sinh (\omega_n \,t)}{\omega_n} \leq\sum_{n=N_{2}+1}^{N_c-1} t e^{- \,\frac{t}{2\varepsilon }}\,\leq B_2 \,\, \frac{t}{\varepsilon•}\,\,e^{- \,\frac{t}{2\varepsilon }}
\end{equation}

Moreover, for all $  n \geq N_c,$ one has  $n \geq \frac{1}{ \varepsilon \sqrt{c}•} ( 1+ \sqrt{1- a \varepsilon c})\Rightarrow\varepsilon \sqrt{c} \,\,{n}^{2}  + a  \sqrt{c} - 2 {n} \geq 0, $ so that $ \frac{2{n}}{ a+\varepsilon n^2•} \leq\sqrt{c}\Leftrightarrow\frac{{n}}{h_{n}•} \leq\sqrt{c}  $ and, hence, $ \omega_n= h_n \sqrt{1-\frac{n^2}{h_n^2•}}   \geq  h_n\sqrt{1-c}\, $


As a  consequence,  if $ 0<\eta<1, $  denoting by  $ \zeta(x)   $   the Riemann zeta function and $\,\,\theta = t/\varepsilon,$ and since $ \varepsilon \,N_c\, \geq 1/\sqrt{c},  $ one has

\begin{equation}
\sum_{n=N_c}^\infty  e^{-h_n t} \,\frac{\sinh (\omega_n \,t)}{\omega_n} \leq \frac{ 2\,\,{c}^{(1-\eta)/2}}{ \sqrt{1-c}•} \,\,\frac {\varepsilon ^{-\eta} \,\,\zeta(1+\eta)}{( 1+ \sqrt{1- a \varepsilon c})^{1-\eta}}\, e^{\,-\frac{t}{2\varepsilon}} 
\end{equation}

\noindent So, there exists a  positive constant $ B_3  $, independent from $ \varepsilon, $such that, taking also  account of (\ref{aa})- (\ref{uu}), it results:
 
\begin{equation}\label{58}
\varepsilon\,|   G(x,\xi,t)|\,\leq  \,B_1 \,( \,t\varepsilon +\varepsilon ^{1/2-k}) \,e^{- \frac{a}{2•}\,t}\, + B_3 \,\, (\varepsilon^{1-\eta}+ \,\varepsilon)  \,e^{-  \frac{\theta }{4\,}}
\end{equation}
 from which,  inequality (\ref{588}) follows.

\section{A priori estimates}

By assuming that   function $ F_0(x), F_1(x)$ and $ F(x,t) $  belong to class $ C^2(\Omega), $   let

\begin{eqnarray}  \label{an73}
\nonumber &\big|\big|\,   F_i^{(j)}\big|\big| = \displaystyle \sup _{ 0\leq\,x\leq L}\, \big|  F_i^{(j)} (x)\big | \quad  \mbox{for}\,\, i= 0,1  \quad  j=0,2
\nonumber \\
 \nonumber &\,|| F\,|| \,= \displaystyle \sup _{ \Omega\,}\, \bigg|  \int_0^t  F(x,\tau ) dt  \,\bigg|,\qquad  || F_{xx}\,|| \,= \displaystyle \sup _{ \Omega\,}\, | \, F_{xx} \,(\,x,\,t\,) \,|.
\end{eqnarray} 

 Besides, indicating by    $ \Gamma(a,z) $ the incomplete  Gamma function and by  $ \zeta(\xi) $  the 
Riemann zeta function, let

\begin{eqnarray}
\nonumber &K_0= || F_0||+\frac{1}{a•} || F_1|| +\frac{1}{a•} ||  F||; \quad  \ K_1 =  B ||\ddot F_0||;
\\
\nonumber\\
&K_{2}=||\ddot F_1|| + a \,\,|| \ddot F_0||\,)\,\,\pi\,\,\,A;   \quad K_3 = || \ddot F_0||\,\zeta(2) \pi \,\quad  H= ||F_{xx}|| \pi A
\\  
\nonumber\\  
\nonumber &H_1 = ||F_{xx}|| \pi \,A \,\,[ 2/a +  4/a^2+ 4+(2/a)^{2-\gamma }• \Gamma(2-\gamma)  + (2/a)^{3-\delta }• \Gamma(3-\delta)]
\nonumber\\
\nonumber\\
&\nonumber k(t)= (2/a)^{2-\gamma }• \Gamma(2-\gamma, ta/2 )
+ (2/a)^{3-\delta }• \Gamma(3-\delta,ta/2)
\end{eqnarray}
 
\noindent where constant $ A $ and $ B  $ are referred  respectively to theorem \ref{th7} and \ref{th9}. Then,   if we  consider that

  \begin{equation}  \label{6666}
 \frac{d }{dt •}  ( e^{-\frac{\varepsilon }{2•}t} )\,\, \sum_{n=1}^\infty \,\int_0^\pi \ddot F_0 \, G_n^o\, \frac{g_n}{n^2•}  d\xi \leq  K_3 \, \varepsilon \,t \,e^{-\frac{a}{2•}t}\,
  \end{equation} 


\noindent where $  G_n^o\,  $ is defined in $(\ref{110})_2$,  according to (\ref{an225}) and  by means of  theorems \ref{th7}, \ref{th8}, \ref{th9}, the following theorem holds:

\begin{theorem} \label{th10}
When data $ (F_1,F_0,F) $ satisfy respectively the hypotheses of  theorem \ref{th:33} \ref{the:dati0}  \ref{t42}, then the following inequality holds:

 \begin{eqnarray}  \label{66}
 &\bigg |u(x,t,\varepsilon ) - e^{-\frac{\varepsilon  t }{2}} U(x,t)\bigg|  \leq  K_0  \bigg (1-  e^{-\frac{\varepsilon  }{2•}t} \bigg )+ \, \varepsilon ^m [H k(t) + H_1 ]+ 
 \nonumber\\ 
  \nonumber & + \varepsilon^m  e^{-\frac{a}{2•}t}\{ K_1 (1+t) +\, K_2 r(t)   +C p(t) + K_3t\}+
   e^{-\frac{t}{4 \varepsilon}}\varepsilon ^m ( C+  K_1+ K_2 ) 
\end{eqnarray}

\noindent where $ m  $  and $ p(t) \, $ are  defined in (\ref{51}), while  $ r(t)$ and $ \, C $ refer to theorem  \ref{th7} and \ref{th8}, respectively.

 \end{theorem}
   
Theorem \ref{th10} allows us to achieve  an estimate characterized  by means of   fast time $ \theta=t/\varepsilon $ and slow time $ \tau= t\varepsilon.  $

\noindent Indeed,   for $ 0< \delta < 1/2, $  since (\ref{a}), one has $ t^{2-\delta}\, e^{-\frac{a}{2•}t} \leq \frac{4(1-2\delta)•}{ae•}^{1-2\delta }\, t^{1+\delta} \, e^{-\frac{a}{4•}t}  . $

 \noindent Moreover, since  $ \forall z>0 $ it results $ \Gamma(a,z)\leq \Gamma(a), $   taking  account  of  (\ref{422}), (\ref{416}), (\ref{417}), 
 (\ref{cc}),
   (\ref{58}) and  (\ref{6666}) ,  it is possible to introduce three positive constants $  {\cal A}, {\cal B},{\cal C}, $ independent from $ \varepsilon,  $  such that it  results:

\begin{eqnarray}  \label{an73}
\nonumber 
 &\bigg |u(x,t,\varepsilon ) - e^{-\frac{\varepsilon t   }{2}\,} U(x,t)\big|\leq   K_1  \bigg (1-  e^{-\frac{\varepsilon   }{2}\,t} \bigg )+ {\cal A} \, \varepsilon ^m +  {\cal B} \,e^{-\frac{t}{4 \varepsilon•}}
\nonumber \\
\nonumber \\
& +  {\cal C}\,[\, \sqrt{\varepsilon t} + (\varepsilon t)^{1-\gamma} + (\varepsilon t)^{1+ \delta}+ t \varepsilon\,] e^{-\frac{a•}{4•}t}    
\end{eqnarray}

  Therefore, indicating by $ K $ a positive constant independent from $ \varepsilon $, when    $ t\in (0,\frac{1}{\varepsilon•}), $ one has: 

\begin{equation}
\bigg|u(x,t,\varepsilon ) - e^{-\frac{\varepsilon t }{2•}\,} U(x,t)\bigg|  \leq  \,\, K 
\end{equation}

\noindent and hence  for slow times $ \tau=\varepsilon t <1, $ the wave is propagated nearly unperturbed.

\noindent Moreover, when $ t>1/\varepsilon, $ damped oscillations prevail.


With regards  to asymptotic properties, it should be remarked that obviously hypotheses on the  source term have an  influence on asymptotic behaviours. If for instance, we suppose that $ F_{xx}(x,t)$  is bounded also when  $ t \rightarrow \infty $ and  $ \int _0^\infty F(x,t) dt <\infty  $ then, we obtain that the   solution $ u(x,t,\varepsilon)  $  is   bounded  when $ t  $ increases to infinity.

\end{document}